\pdfoutput=1
\documentclass[superscriptaddress,aps,pra,10pt,twocolumn,showpacs,nofootinbib,longbibliography]{revtex4-1}
\usepackage{amsmath,amssymb,amsthm}
\usepackage{amsmath,kbordermatrix,blkarray}
\usepackage[colorlinks=true,citecolor=blue,urlcolor=blue]{hyperref}
\usepackage[pdftex]{graphicx}
\usepackage{times,txfonts}
\usepackage{braket}
\usepackage{color}
\usepackage{mathtools}
\usepackage{natbib}
\usepackage{slashbox}

\newcommand{\be}{\begin{equation}}
\newcommand{\ee}{\end{equation}}
\newcommand{\ba}{\begin{eqnarray}}
\newcommand{\ea}{\end{eqnarray}}

\newcommand{\tr}{\operatorname{Tr}}

\newtheorem{definition}{Definition}
\newtheorem{thm}{Theorem}
\newtheorem{cor}{Corollary}

\begin{document}
	
\title{Operational characterization of quantumness of unsteerable bipartite states }

\author{Debarshi Das}
 \email{debarshidas@jcbose.ac.in}
\affiliation{Centre for Astroparticle Physics and Space Science (CAPSS), Bose Institute, Block EN, Sector V, Salt Lake, Kolkata 700 091, India}
		
	\author{Bihalan Bhattacharya}
	\email{bihalan@gmail.com}
	\affiliation{S. N. Bose National Centre for Basic Sciences, Block JD, Sector III, Salt Lake, Kolkata 700 098, India}
		
	\author{Chandan Datta}
	\email{chandan@iopb.res.in}
	\affiliation{Institute of Physics, Sachivalaya Marg, Bhubaneswar 751005, Odisha, India}
	\affiliation{Homi Bhabha National Institute, Training School Complex, Anushakti Nagar, Mumbai 400085, India}
	\author{Arup Roy}
	\email{arup145.roy@gmail.com}
	\affiliation{Physics and Applied Mathematics Unit, Indian Statistical Institute, 203 B. T. Road, Kolkata 700108, India}
	
	\author{C. Jebaratnam}
	\email{jebarathinam@gmail.com}
	\affiliation{S. N. Bose National Centre for Basic Sciences, Block JD, Sector III, Salt Lake, Kolkata 700 098, India}
	
	\author{A. S. Majumdar}
	\email{archan@bose.res.in}
	\affiliation{S. N. Bose National Centre for Basic Sciences, Block JD, Sector III, Salt Lake, Kolkata 700 098, India}
	
	\author{R. Srikanth}
	\email{srik@poornaprajna.org}
	\affiliation{Poornaprajna Institute of Scientific Research Bangalore- 560 080, Karnataka, India}

\begin{abstract}
Recently, the  quantumness of local  correlations arising  from separable  states
in the context  of a
Bell  scenario  has  been studied  and
linked with  superlocality [\href{https://journals.aps.org/pra/abstract/10.1103/PhysRevA.95.032120}{Phys. Rev. A {\bf 95}, 032120 (2017)}].
Here we investigate the quantumness of unsteerable
correlations   in  the   context   of  a   given  steering   scenario.
Generalizing the concept  of superlocality, we define
  as   \textit{super-correlation},  the   requirement  for   a  larger
  dimension of  the preshared randomness to  simulate the correlations
  than  that  of  the  quantum   states  that  generate  them.   Since
  unsteerable states  form a  subset of  Bell local  states, it  is an
  interesting   question   whether    certain   unsteerable states  can   be
  super-correlated.  Here, we answer this question in the affirmative.
In particular, the quantumness  of certain unsteerable correlations
can be pointed out by the notion of \textit{super-unsteerability}, the
requirement for a larger dimension  of the classical variable that the
steering party has  to preshare with the trusted  party for simulating
the correlations than that of the quantum states which reproduce them.
This provides  a generalized approach  to quantify the  quantumness of
unsteerable correlations in convex operational theories.
\end{abstract} 
	
	
	\maketitle
	
\section{INTRODUCTION} 
Ideas and concepts of classical physics significantly differ from that of quantum mechanics (QM). A pioneering contribution showing an incompatibility between local-realism (which is a classical concept) and QM is the Bell-CHSH (Bell-Clauser-Horne-Shimony-Holt) inequality \cite{Bell, chsh, bell2}, which shows that measurements on certain spatially separated system can lead to nonlocal correlations which cannot be explained by local hidden variable (LHV) theory. Bell-CHSH inequality puts an upper bound on the correlations admitting any LHV model. Violation of Bell-CHSH inequality is though, not a defining nonclassical feature of QM as there are post-quantum correlations obeying the no-signalling (NS) principle, which also violate a Bell inequality. Nonlocality in QM is limited by the Tsirelson bound \cite{tsi}. Motivated by this fact, Popescu and Rohrlich proposed NS correlations which are more nonlocal than quantum correlations \cite{pr}.\\

In generalized NS theory, the only constraint on the correlations is the NS principle \citep{barrett, masanes}. The set of NS correlations form a polytope which can be  categorized by nonlocal and local vertices in contrast to the set of correlations arising out from QM which form convex set but fail to form a polytope \cite{pp}. Since QM correlations are contained within the NS polytope, any QM correlation can be written as a convex combination of the extremal boxes of the NS polytope. One of the goals of studying generalized NS theory is to find out how one can single out QM from other NS theories \cite{bell2, ps, ps2}.\\

The seminal argument by Einstein, Podolsky and Rosen (EPR) \cite{epr} to demonstrate the incompleteness of QM, motivated Schrodinger with the concept of `quantum steering' \cite{scro}. The concept of steering in the form of a task has been introduced recently \cite{steer, steer2}. The task of quantum steering is to prepare different ensembles at one part of a bipartite system by performing local quantum measurements on another part of the bipartite system in such a way that these ensembles cannot be explained by a local hidden state (LHS) model. This implies that the steerable correlations cannot be reproduced by a local hidden variable-local hidden state (LHV-LHS) model. In recent years, investigations related to quantum steering have been acquiring considerable significance, as evidenced by a wide range of studies \cite{st8, st10, steer22, steer3, st4, st9, st5, steer24, s6, new}. \\

Bell-nonlocal states form a subset of  steerable states which also form a subset of entangled states \cite{steer, st11}. However, unlike quantum nonlocality \cite{bell2} and entanglement \cite{ent}, the task of quantum steering is inherently asymmetric \cite{st7}. In this case, the outcome statistics of one subsystem (which is being ‘steered’) is due to valid QM measurements on a valid QM state. On the other hand, there is no such constraint for the other subsystem. The study of quantum steering also finds applications in semi device independent scenario where the party, which is being ‘steered’, has trust on his/her quantum device but the other party's device is untrusted. Secure quantum key distribution (QKD) using quantum steering has been demonstrated \cite{st12}, where one party cannot trust his/her devices.\\

Quantum discord \cite{disc, disc2, disc3} and local broadcasting \cite{bc} indicate the existence of quantumness even in separable states, and these concepts can be linked with the non-commutativity of measurements \cite{disc4}. On the other hand, quantum nonlocality and steering are also associated with incompatibility of measurements \cite{com1, com2, com, com3}. From an operational perspective, nonlocal or steerable states require pre-shared randomness with non-zero communication cost \cite{cc1, cc2}. Here we are concerned with the question of how to give such an operational characterization to the quantumness of local or unsteerable correlations. \\

Such a question has been partially addressed in case of local correlations. It has been demonstrated that there exists some local correlations for which the dimension of the pre-shared randomness required to simulate the correlations exceeds the dimension of the quantum system reproducing them by applying suitable measurements. This is known as superlocality \cite{sl1, sl2, sl3, sl4, sl5, sl6}. On the other hand, any unsteerable correlation can be reproduced by appropriate
measurements on an appropriate separable state \cite{nsst, guhne}.  Moreover,  Moroder  et. al.
\cite{guhne} have shown that the unsteerable correlations arising from
$2  \times 2  \times 2$  experimental  scenario ($2$ parties, $2$ measurement settings per party, $2$ outcomes per measurement setting) can  be reproduced  by
classical-quantum states (which form a  subset of the set of separable
  states) with dimension at the untrusted party $d \leq 4$.\\

In this work,  our motivation is to  analyze  the
  resource requirement for simulating  unsteerable correlations in the
  context  of a  given steering  scenario.  Extending  the concept  of
  superlocality,   we   define   \textit{super-correlation}   as   the
  requirement for  a larger dimension  of the preshared  randomness to
  simulate  the correlations  than  that of  the  quantum states  that
  generate them.   Superlocality is  an instance  of super-correlation.
  Since unsteerable states  form a  subset of  Bell local  states, it  is an
  interesting   question   whether    certain   unsteerable states  can   be
  super-correlated. Here,  we  find that  certain unsteerable  correlations
  evidence     super-correlation,      a     phenomenon      we     term
  ``super-unsteerability''.   In  other  words,  we  show  that
  quantumness   is   necessary   to  reproduce   certain   unsteerable
  correlations in  the scenario  where the  dimension of  the resource
  reproducing the correlations is restricted.
More specifically, we show that there are certain unsteerble correlations in the $2 \times 2 \times 2$ experimental scenario ($2$ parties, $2$ measurement settings per party, $2$ outcomes per measurement setting) whose simulation with LHV-LHS model requires the steering party to preshare hidden variables with dimension exceeding the local Hilbert space dimension of the quantum systems (generating the given unsteerable correlation) at the steering party's side. This is termed as ``super-unsteerability".\\

The plan of the paper is as follows. 
 In Section II the basic notions of NS polytope and the fundamental ideas of quantum steering has been presented. Our purpose is to decompose the given NS correlation
in terms of convex combinations of extremal boxes of NS polytope which
leads  to a  LHV-LHS  decomposition  of the  given  correlation \cite{DDJ+17}. In Section III, we present formal definition of super-unsteerability and demonstrate some specific examples of it. In Section IV, we illustrate how quantumness is captured by the notion of super-unsteerability. The inequivalence between superlocality and super-unsteerability has been demonstrated in Section V. Finally, in the concluding Section VI, we elaborate a bit on the significance of the results obtained.

\section{Framework}

\subsection{No-signalling Boxes}

In this work, we are interested in generalized NS bipartite correlations which are treated as “boxes” shared between two parties, say Alice and Bob. The input variables on Alice's and Bob's sides are denoted by $x$ and $y$ respectively, and the outputs are denoted by $a$ and $b$ respectively. We restrict ourselves to the probability space in which the boxes have binary inputs and binary outputs, i.e., $x, y, a, b \in \{0, 1\}$. In this case, the state of every box is given by the set of 16 joint probability distributions $p(ab|xy)$.  A bipartite box $P$ = $P(ab|xy)$ := $\{ p(ab|xy) \}_{a,x,b,y}$ is the set of joint probability distributions $p(ab|xy)$ for all possible $a$, $x$, $b$, $y$. The single-partite   
box $P(a|x)$ := $\{p(a|x)\}_{a,x}$ of a NS box $P(ab|xy)$ 
is the set of marginal probability distributions $p(a|x)$ for all possible $a$ and $x$; which are given by,
\be
p(a|x)=\sum_b p(ab|xy), \quad \forall a,x,y.
\ee
The single-partite box $P(b|y)$ := $\{p(b|y)\}_{b,y}$ of a NS box $P(ab|xy)$ is the set of marginal
probability distributions 
$p(b|y)$ for all possible $b$ and $y$; which are given by,
\be
p(b|y)=\sum_a p(ab|xy), \quad \forall x,b,y.
\ee
\\

A NS box  $P(ab|xy)$ is  nonlocal if it cannot be reproduced by a LHV model, 
\begin{equation}
p(ab|xy)=\sum_\lambda p(\lambda) p(a|x,\lambda)p(b|y,\lambda) \hspace{0.3cm} \forall a,b,x,y;
\end{equation}
where $\lambda$ denotes shared randomness which occurs with probability $p(\lambda)$; each $p(a|x,\lambda)$ and $p(b|y,\lambda)$ are conditional probabilities. The set of local boxes which have a LHV model forms a convex polytope called local 
polytope.
In the case of two-binary-inputs and two-binary-outputs Bell scenario, the local polytope has $16$ extremal boxes 
which are local-deterministic boxes given by,
\be \label{LDB}
    P_{D}^{\alpha \beta \gamma \epsilon} (ab|xy) = 
\begin{dcases}
    1,& \text{if } a = \alpha x \oplus \beta, b = \gamma y \oplus \epsilon \\
    0,              & \text{otherwise}.
\end{dcases}
\ee
Here, $\alpha, \beta, \gamma, \epsilon \in \{0,1\}$ and $\oplus$ denotes addition modulo $2$. Any local box can be written
as a convex mixture of the local-deterministic boxes. All the local-deterministic boxes as defined above
can be written as the product of marginals corresponding to Alice and Bob,
i.e., $P_D^{\alpha\beta\gamma\epsilon}(ab|xy)=P^{\alpha\beta}_D(a|x)P^{\gamma\epsilon}_D(b|y)$,
with the  deterministic box on Alice's side given by
\begin{equation}
P_D^{\alpha\beta}(a|x)=\left\{
\begin{array}{lr}
1, & a=\alpha x\oplus \beta\\
0 , & \text{otherwise}\\
\end{array}
\right. 
\label{}
\end{equation} 
and the  deterministic box on Bob's side given by,
\begin{equation}
P_D^{\gamma\epsilon}(b|y)=\left\{
\begin{array}{lr}
1, & b=\gamma x\oplus \epsilon\\
0 , & \text{otherwise}.\\
\end{array}
\right.   
\label{}
\end{equation}

A local box satisfies the complete set of Bell  inequalities \cite{werner}. In the case 
of  two-binary-inputs and two-binary-outputs, the Bell-CHSH inequalities \cite{chsh} are given by,
\begin{eqnarray}
\label{chsh}
\mathcal{B}_{\alpha \beta \gamma} =&& (-1)^{\gamma} \langle A_0 B_0 \rangle + (-1)^{\beta \oplus \gamma} \langle A_0 B_1 \rangle \nonumber\\
&&+ (-1)^{\alpha \oplus \gamma} \langle A_1 B_0 \rangle + (-1)^{\alpha \oplus \beta \oplus \gamma \oplus 1} \langle A_1 B_1 \rangle \leq 2,
\end{eqnarray}
where $\alpha, \beta, \gamma \in  \{0, 1 \}$,   $\langle A_x B_y \rangle =  \sum_{a,b} (-1)^{a\oplus b} p(ab|x y)$,
form the complete set of Bell inequalities.
All these tight Bell inequalities form the nontrivial facets for the local polytope. 
All nonlocal boxes lie outside the local polytope and violate a Bell inequality.

The set of all bipartite  two-input-two-output NS boxes forms an $8$ dimensional convex polytope with $24$ extremal boxes 
\cite{barrett}, which can be divided into two classes:  i) nonlocal boxes having $8$ Popescu-Rohrlich (PR) boxes as extremal boxes, which are given by,
\be
    P_{PR}^{\alpha \beta \gamma} (ab|xy) = 
\begin{dcases}
    \frac{1}{2},& \text{if } a \oplus b = x.y \oplus \alpha x \oplus \beta y \oplus \gamma \\
    0,              & \text{otherwise},
\end{dcases}
\ee
and  ii) local boxes having $16$ local-deterministic boxes as extremal boxes, which are given in Eq. (\ref{LDB}). The extremal boxes in a given class are equivalent under ``local reversible operations" (LRO). By using LRO Alice and Bob can convert any extremal box in one class into any other extremal box within the same class. LRO is designed \cite{barrett} as follows: Alice may relabel her inputs: $x \rightarrow x \oplus 1$, and she may relabel her outputs (conditionally on the input) : $a \rightarrow a \oplus \alpha x \oplus \beta$; Bob can perform similar operations.

\subsection{Quantum Steering}

Let us consider a steering scenario where two spatially seperated parties, say Alice and Bob, share an unknown quantum system $\rho_{AB}\in \mathcal{B}(\mathcal{H}_A \otimes \mathcal{H}_B)$. Here $\mathcal{B}(\mathcal{H}_A \otimes \mathcal{H}_B)$ stands for the set of all bounded linear operators acting on the Hilbert space $\mathcal{H}_A \otimes \mathcal{H}_B$.  Alice 
performs a set of black-box measurements  and the Hilbert-space dimension of Bob's 
subsystem is known. Such a scenario is called one-sided device-independent since
Alice's measurement operators $\{M_{a|x}\}_{a,x}$ are unknown. The steering scenario 
is completely characterized by the set of unnormalized conditional states on Bob's
side $\{\sigma_{a|x}\}_{a,x}$, which is called an unnormalized assemblage. Each element
in the unnormalized assemblage is given by $\sigma_{a|x}=p(a|x)\rho_{a|x}$,  where $p(a|x)$ is the conditional probability of getting the outcome $a$ when Alice performs the measurement $x$;
$\rho_{a|x}$ is the normalized conditional state on Bob's side.
Quantum theory predicts that all valid assemblages should satisfy the following criteria:
\begin{equation}
\sigma_{a|x}=\tr_A ( M_{a|x} \otimes \openone \rho_{AB}) \hspace{0.5cm} \forall \sigma_{a|x} \in \{\sigma_{a|x}\}_{a,x}
\end{equation}
Let $\Sigma^{S}$ denote the set of all valid assemblages. \\

In the above scenario, Alice demonstrates steerability to Bob if the assemblage does not have a local hidden state (LHS) model, i.e., if for all $a$, $x$, there is no decomposition of $\sigma_{a|x}$ in the form,
\begin{equation}
\sigma_{a|x}=\sum_\lambda p(\lambda) p(a|x,\lambda) \rho_\lambda,
\end{equation}
where $\lambda$ denotes classical random variable which occurs with probability 
$p(\lambda)$; $\rho_{\lambda}$ are called local hidden states which satisfy $\rho_\lambda\ge0$ and
$\tr\rho_\lambda=1$. Let $\Sigma^{US}$ denote the set of all unsteerable assemblages. 
Any element in the given assemblage $\{\sigma_{a|x}\}_{a,x} \in \Sigma^{US}$ can be decomposed
in terms of deterministic distributions  as follows:
\begin{equation}
\sigma_{a|x}=\sum_\chi D(a|x,\chi) \sigma_\chi,
\end{equation}
where $D(a|x,\chi):=\delta_{a,f(x,\chi)}$ is the single-partite
extremal conditional probability for Alice determined by the variable 
$\chi$ through the function $f(x,\chi)$ and $\sigma_\chi$ satisfy
$\sigma_\chi\ge0$ and $\tr \sum_{\chi} \sigma_\chi=1$ \cite{newpusey}.\\

Suppose Bob performs a set of projective measurements $\{\Pi_{b|y}\}_{b,y}$ on $\{\sigma_{a|x}\}_{a,x}$. Then the scenario is characterized by the set of measurement correlations, or box
between Alice and Bob $P(ab|xy)$:=$\{p(ab|xy)\}_{a,x,b,y}$, where $p(ab|xy)$ = $\tr ( \Pi_{b|y} \sigma_{a|x} )$. The box $P(ab|xy)$ detects steerability from Alice to Bob, 
iff it does not have a decomposition as follows \cite{steer, steer2}: 
\begin{equation}
p(ab|xy)= \sum_\lambda p(\lambda) p(a|x,\lambda) p(b|y, \rho_\lambda) \hspace{0.3cm} \forall a,x,b,y; \label{LHV-LHS}
\end{equation}
where, $\sum_{\lambda} p(\lambda) = 1$, $p(a|x, \lambda)$ denotes an arbitrary probability distribution arising from local hidden variable (LHV) $\lambda$ ($\lambda$ occurs with probability $p(\lambda)$) and $p(b|y, \rho_{\lambda}) $ denotes the quantum probability of outcome $b$ when measurement $y$ is performed on local hidden state (LHS) $\rho_{\lambda}$. Hence, the box $P(ab|xy)$ will be called steerable iff it does not have a LHV-LHS model. In a given steering scenario, correlations having LHV-LHS model form a convex subset of the set of all correlations in that scenario.\\

 Till date various criteria for showing quantum steering have been demonstrated \cite{stt, stt2, stt3, stt5}, but none of these criteria is a necessary and sufficient condition for quantum steering. Only recently, a necessary and sufficient condition for quantum steering in the $2 \times 2 \times 2$ experimental scenario with mutually unbiased measurements at trusted party has been established \cite{stt6}. Suppose two spatially separated parties Alice and Bob each have a choice between two dichotomic measurements to perform: $\{ A_1, A_2\}$, $\{ B_1, B_2 \}$, and outcomes of $A$ are labeled $a \in \{0, 1 \}$ and similarly for the other measurements. Furthermore, suppose that $B_1$ and $B_2$ are two mutually unbiased measurements. In this scenario, the necessary and sufficient condition for quantum steering from Alice to Bob is given by,
\begin{align}
&\sqrt{\langle (A_1 + A_2) B_1 \rangle^2 + \langle (A_1 + A_2) B_2 \rangle^2 } \nonumber \\
&+\sqrt{\langle (A_1 - A_2) B_1 \rangle^2 + \langle (A_1 - A_2) B_2 \rangle^2 } \leq 2, \label{chshst}
\end{align}
where $\langle A_x B_y \rangle =  \sum_{a,b} (-1)^{a\oplus b} p(ab|xy)$. This inequality is called the analogous CHSH inequality for quantum steering.\\

Now, we are in a position to establish our results, i.e., demonstrating the notion of \textit{``super-unsteerability"} for certain unsteerable correlations.

\section{Super-unsteerability}

In this Section we are going to present the formal definition of the notion \textit{``super-unsteerability"} which is followed by some of its examples. Before that we present the definition of \textit{``superlocality"} \cite{sl1, sl2, sl3, sl4, sl5, sl6}. Consider the Bell scenario, where both parties perform black box measurements. In this scenario, superlocality is defined as follows:
\begin{definition}
Suppose we have a quantum state in $\mathbb{C}^{d_A}\otimes\mathbb{C}^{d_B}$
and measurements which produce a local bipartite box $P(ab|xy)$ := $\{ p(ab|xy) \}_{a,x,b,y}$.
Then, superlocality holds iff there is no decomposition of the box in the form,
\begin{equation}
p(ab|xy)=\sum^{d_\lambda-1}_{\lambda=0} p(\lambda) p(a|x, \lambda) p(b|y, \lambda) \hspace{0.3cm} \forall a,x,b,y,
\end{equation}
with dimension of the shared randomness/hidden variable $d_\lambda\le$ min($d_A$, $d_B$).  Here $\sum_{\lambda} p(\lambda) = 1$, $p(a|x, \lambda)$ and $p(b|y, \lambda)$ denotes arbitrary probability distributions arising from LHV $\lambda$ ($\lambda$ occurs with probability $p(\lambda)$).
\end{definition}

Now, consider a different scenario where one of the parties (say, Alice) performs black box measurements and another party (say, Bob) performs quantum measurements. In this steering scenario, we define the notion of \textit{``super-unsteerability"} as follows:
\begin{definition}
Suppose we have a quantum state in $\mathbb{C}^{d_A}\otimes\mathbb{C}^{d_B}$
and measurements which produce a unsteerable bipartite  box $P(ab|xy)$ := $\{ p(ab|xy) \}_{a,x,b,y}$.
Then, super-unsteerability holds iff there is no decomposition of the box in the form,
\begin{equation}
p(ab|xy)=\sum^{d_\lambda-1}_{\lambda=0} p(\lambda) p(a|x, \lambda) p(b|y, \rho_{\lambda}) \hspace{0.3cm} \forall a,x,b,y,
\end{equation}
with dimension of the shared randomness/hidden variable $d_\lambda\le d_A$.  Here $\sum_{\lambda} p(\lambda) = 1$, $p(a|x, \lambda)$ denotes an arbitrary probability distribution arising from LHV $\lambda$ ($\lambda$ occurs with probability $p(\lambda)$) and $p(b|y, \rho_{\lambda}) $ denotes the quantum probability of outcome $b$ when measurement $y$ is performed on LHS $\rho_{\lambda}$ in $\mathbb{C}^{d_B}$.
\end{definition}
Hence, in order to demonstrate super-unsteerability of a given unsteerable correlation, we have to consider a LHV-LHS model of the given correlation with minimum dimension of the shared randomness and have to check whether this minimum dimension is greater than the local Hilbert space dimension of the shared quantum system (reproducing the given unsteerable correlation) at the untrusted party's side (who steers the other party, in the present case Bob). In the following we describe the procedure adopted in the present study to minimize the dimension of the shared randomness associated with the LHV-LHS model of the given unsteerable correlation.\\

At first we decompose the given unsteerable correlation (which is local as well) in terms of local-deterministc boxes. Then this local deterministic boxes are written as the product of marginals corresponding to the two parties. This decomposition produces a classical simulation protocol with LHV model of the given unsteerable correlation. In order to reduce the dimension of the shared randomness in this decomposition of the given correlation, we make each probability distribution at Bob's side non-deterministic and keep each probability distribution at Alice's side as deterministic. If each non-deterministic distribution at Bob's side can be produced by performing quantum measurements in the given steering scenario from some pure state, then this decomposition gives a LHV-LHS model in the given steering scenario. However, the dimension of the shared randomness in the above decomposition of the given correlation may be further reduced if there are several sets of equal non-determinstic probability distributions (which have some quantum realisation as described earlier) at Bob's side. In this case, by taking each of the equal non-determinstic probability distributions at Bob's side as common and by making corresponding  probability distributions at Alice's side non-deterministic, the dimension of the shared randomness can be further reduced. This minimizes the dimension of the shared randomness in the above decomposition of the unsteerable correlation with different probability distributions (deterministic/non-deterministic) at Alice's side and with non-deterministic probability distributions having quantum realisations at Bob's side.

\subsection{Super-unsteerability: Example 1}

Consider the white noise-BB84 family defined as

\begin{equation}
\label{bb84}
P_{BB84}(ab|xy) = \frac{1 + (-1)^{a \oplus b \oplus x.y} \delta_{x,y} V }{4},
\end{equation}
where $V$ is a real number such that $0 < V \leq 1$; $x$, $y$ denote the input variables on Alice's and Bob's sides respectively; and $a$, $b$ denote the outputs on Alice's and Bob's sides respectively. We restrict ourselves to the probability space in which the boxes have binary inputs and binary outputs, i.e., $x, y, a, b \in \{0, 1\}$. The above box  is local as it does not violate a Bell-CHSH inequality (\ref{chsh}). Therefore, it can be reproduced by sharing classical randomness.
We now give an example of simulation of the white noise-BB84 family by using a quantum state which has quantumness. 
Consider that the two spatially separated parties (say, Alice and Bob) share the two qubit Werner state,
\begin{equation}
\label{w}
\rho_V = V | \psi^- \rangle \langle \psi^-| + \frac{1-V}{4} \mathbb{I}_4,
\end{equation}
where, $|\psi^- \rangle = \frac{1}{\sqrt{2}} (|01 \rangle - |10 \rangle)$ ($|0\rangle$ and $|1\rangle$ are the eigenstates of $\sigma_z$), $\mathbb{I}_4$ is the $4 \times 4$ identity matrix and $0 < V \leq 1$. The above states are entangled iff $V > \frac{1}{3}$. The Werner states $\rho_V$ have nonzero quantumness (as quantified by quantum discord \cite{disc, disc2, disc3, disc4}) for any $V>0$. The white noise-BB84 family can be produced from the 
two-qubit Werner state if Alice performs the projective measurements of observables corresponding to the operators $A_0 = - \sigma_z$ and $A_1 = \sigma_x$, and Bob performs projective measurements of observables corresponding to the operators $B_0 =  \sigma_z$ and $B_1 = \sigma_x$.\\ 

The BB84 family (\ref{bb84}) violates the analogous CHSH inequality for steering (\ref{chshst}) 
iff $V > \frac{1}{\sqrt{2}}$. Therefore, in this range the white noise-BB84 family detects steering in the $2 \times 2 \times 2$ experimental scenario where Alice performs black-box (uncharacterized) measurements and Bob performs two mutually unbiased qubit measurements. For instance, it detects steerability of the two-qubit Werner state in this range. Because the two-qubit Werner is steerable
in the above steering scenario iff $V>1/\sqrt{2}$. In the following, we will demonstrate that for $0<V\le1/\sqrt{2}$, the BB84 box demonstrates super-unsteerability in the $2 \times 2 \times 2$ steering scenario.

\subsubsection*{Simulating unsteerable white noise-BB84 family with LHV at one side and LHS at another side}

In the context of no-signaling polytope, the  white noise-BB84 distribution can be decomposed as
\begin{equation}
P_{BB84}(ab|xy) = V \bigg(\frac{P_{PR}^{000} + P_{PR}^{110}}{2}\bigg) + (1-V) P_N, \label{bb84box}
\end{equation}
where $P_N$ is the maximally mixed box, i.e., $P_N(a b| x y) = \frac{1}{4} \forall a, b, x, y$. We obtain
\begin{eqnarray}
P_{BB84}(ab|xy) = &&2V \frac{1}{2} \bigg(\frac{1}{2} P_{PR}^{000} + \frac{1}{2} P_N \bigg) + 2V \frac{1}{2} \bigg(\frac{1}{2} P_{PR}^{110} + \frac{1}{2} P_N \bigg)\nonumber\\&&  + (1-2V) P_N .
\end{eqnarray}
Each box in the above decomposition can be decomposed in terms of the local deterministic boxes as follows:
\begin{equation}
\frac{1}{2} P_{PR}^{000} + \frac{1}{2} P_N = \frac{1}{8} \sum_{\alpha, \beta, \gamma} P_D^{\alpha \beta \gamma (\alpha \gamma \oplus \beta)} (a b|x y),
\end{equation}
\begin{equation}
\frac{1}{2} P_{PR}^{110} + \frac{1}{2} P_N =  \frac{1}{8} \sum_{\alpha, \beta, \gamma} P_D^{\alpha \beta \gamma (\bar{ \alpha}             \bar{\gamma} \oplus \beta)} (a b|x y),
\end{equation}

where $\bar{ \alpha} = \alpha \oplus 1$, $\bar{\gamma}= \gamma \oplus 1$;
 and, 
\begin{equation}
P_N = \frac{1}{16} \sum_{\alpha, \beta, \gamma, \epsilon} P_{D}^{\alpha \beta \gamma \epsilon}(ab|xy).
\end{equation}
Using the above decompositions and the relation $P_{D}^{\alpha \beta \gamma \epsilon} (ab|xy) = P_D^{\alpha \beta}(a|x) P_D^{\gamma \epsilon}(b|y)$ one obtains
\begin{align}
P_{BB84}&(ab|xy)\nonumber\\  &= \frac{1}{4} P_D^{00} \bigg[ 2V \bigg( \frac{P_D^{00} + P_D^{10} + P_D^{01} + P_D^{10}}{4} \bigg) \nonumber\\
 &\quad+ (1-2V) \bigg( \frac{P_D^{00} + P_D^{10} + P_D^{01} + P_D^{11}}{4} \bigg) \bigg] \nonumber\\
&\quad+  \frac{1}{4} P_D^{01} \bigg[ 2V \bigg( \frac{P_D^{01} + P_D^{11} + P_D^{00} + P_D^{11}}{4} \bigg) \nonumber\\
&\quad+ (1-2V) \bigg( \frac{P_D^{00} + P_D^{10} + P_D^{01} + P_D^{11}}{4} \bigg) \bigg] \nonumber \\
&\quad+  \frac{1}{4} P_D^{10} \bigg[ 2V \bigg( \frac{P_D^{00} + P_D^{11} + P_D^{00} + P_D^{10}}{4} \bigg) \nonumber\\
&\quad+ (1-2V) \bigg( \frac{P_D^{00} + P_D^{10} + P_D^{01} + P_D^{11}}{4} \bigg) \bigg] \nonumber\\
&\quad+ \frac{1}{4} P_D^{11} \bigg[ 2V \bigg( \frac{P_D^{01} + P_D^{10} + P_D^{01} + P_D^{11}}{4} \bigg)  \nonumber\\
&\quad+ (1-2V) \bigg( \frac{P_D^{00} + P_D^{10} + P_D^{01} + P_D^{11}}{4} \bigg) \bigg] \nonumber\\
&= \sum_{\lambda=0}^{3} p(\lambda) P(a|x, \lambda) P(b|y, \rho_{\lambda}) ,
\label{bblhvlhs}
\end{align}
where, $P(a|x, \lambda)$ := $\{p(a|x, \lambda)\}_{a,x}$ is the set of conditional probabilities $p(a|x,\lambda)$ for all possible $a$ and $x$; $P(b|y, \rho_{\lambda})$ := $\{p(b|y, \rho_{\lambda})\}_{b,y}$ is the set of conditional probabilities $p(b|y,\rho_{\lambda})$ for all possible $b$ and $y$.\\
In the decomposition (\ref{bblhvlhs}) $p(0)$ = $p(1)$ = $p(2)$ = $p(3)$ = $\frac{1}{4}$, and \\
$P(a|x,0)$ = $P_D^{00}$, $P(a|x,1)$ = $P_D^{01}$, $P(a|x,2)$ = $P_D^{10}$, $P(a|x,3)$ = $P_D^{11}$. Now let us set
\begin{equation}      
P(b|y,\rho_0) =\begin{tabular}{c|cc}
 \backslashbox{(y)}{(b)} & (0) & (1) \\\hline
(0) & $\frac{1+V}{2}$ & $\frac{1-V}{2}$  \\
(1) & $\frac{1-V}{2}$ & $\frac{1+V}{2}$  \\
\end{tabular}=\langle \psi _0 | \{\Pi_{b|y}\}_{b,y} | \psi_0 \rangle,
\end{equation}
where each row and column corresponds to a fixed measurement $(y)$ and a fixed outcome $(b)$ respectively. This convention is presented in \cite{not}. Throughout the paper we will follow the same convention. We set the other $P(b|y,\rho_{\lambda})$'s as 
\begin{eqnarray}
&&P(b|y,\rho_1) = \begin{tabular}{c|cc}
 \backslashbox{(y)}{(b)} & (0) & (1) \\\hline
(0) & $\frac{1-V}{2}$ & $\frac{1+V}{2}$  \\
(1) & $\frac{1+V}{2}$ & $\frac{1-V}{2}$  \\ 
\end{tabular}=\langle \psi _1 | \{\Pi_{b|y}\}_{b,y} | \psi_1 \rangle,\\
&&P(b|y,\rho_2) =\begin{tabular}{c|cc}
 \backslashbox{(y)}{(b)} & (0) & (1) \\\hline
(0) & $\frac{1+V}{2}$ & $\frac{1-V}{2}$  \\
(1) & $\frac{1+V}{2}$ & $\frac{1-V}{2}$  \\ 
\end{tabular}=\langle \psi _2 | \{\Pi_{b|y}\}_{b,y} | \psi_2 \rangle,\\
&&P(b|y,\rho_3) =\begin{tabular}{c|cc}
 \backslashbox{(y)}{(b)} & (0) & (1) \\\hline
(0) & $\frac{1-V}{2}$ & $\frac{1+V}{2}$  \\
(1) & $\frac{1-V}{2}$ & $\frac{1+V}{2}$  \\ 
\end{tabular}=\langle \psi _3 | \{\Pi_{b|y}\}_{b,y} | \psi_3 \rangle,
\end{eqnarray}
where $\{\Pi_{b|y}\}_{b,y}$ corresponds to the set of projective measurements of two observables corresponding to the operators $B_0 = |\uparrow_0 \rangle \langle \uparrow_0|$ $-$ $|\downarrow_0 \rangle \langle \downarrow_0|$ and $B_1 =  |\uparrow_1 \rangle \langle \uparrow_1|$ $-$ $|\downarrow_1 \rangle \langle \downarrow_1|$, here $\{ |\uparrow_0 \rangle$, $|\downarrow_0 \rangle \}$ is an arbitrary orthonormal basis in the Hilbert space $\mathcal{C}^2$ and the orthonormal basis $\{ |\uparrow_1 \rangle, |\downarrow_1 \rangle \}$ in the Hilbert space $\mathcal{C}^2$ is such that aforementioned two measurements define two arbitrary projective mutually unbiased measurements in the Hilbert space $\mathcal{C}^2$.
The quantum states $|\psi_{\lambda} \rangle$ in the Hilbert space $\mathcal{C}^2$ that produce $P(b|y,\rho_{\lambda})$s ($\lambda=0,1,2,3$), are given as follows:
\begin{equation}
|\psi_0 \rangle = \sqrt{\frac{1+V}{2}} |\uparrow_0 \rangle + e^{i \phi_0} \sqrt{\frac{1-V}{2}} |\downarrow_0 \rangle,
\end{equation}
where $cos \phi_0 = - \frac{V}{\sqrt{1+V}\sqrt{1-V}}$,
\begin{equation}
|\psi_1 \rangle = \sqrt{\frac{1-V}{2}} |\uparrow_0 \rangle + e^{i \phi_1} \sqrt{\frac{1+V}{2}} |\downarrow_0 \rangle,
\end{equation}
where $cos \phi_1 =  \frac{V}{\sqrt{1+V}\sqrt{1-V}}$. 
\begin{equation}
|\psi_2 \rangle = \sqrt{\frac{1+V}{2}} |\uparrow_0 \rangle + e^{i \phi_2} \sqrt{\frac{1-V}{2}} |\downarrow_0 \rangle,
\end{equation}
where $cos \phi_2 =  \frac{V}{\sqrt{1+V}\sqrt{1-V}}$, 
\begin{equation}
|\psi_3 \rangle = \sqrt{\frac{1-V}{2}} |\uparrow_0 \rangle + e^{i \phi_3} \sqrt{\frac{1+V}{2}} |\downarrow_0 \rangle,
\end{equation}
where $cos \phi_3 = - \frac{V}{\sqrt{1+V}\sqrt{1-V}}$. Now, $|cos \phi_i| \leq 1$ ($i = 0,1,2,3$) implies that $V \leq \frac{1}{\sqrt{2}}$.\\

Hence, the LHV-LHS decomposition of  $P_{BB84}(ab|xy)$ for $ V \leq \frac{1}{\sqrt{2}}$ can be realized with hidden variable having dimension $4$ (with two arbitrary projective mutually unbiased measurements at trusted party).

\begin{thm}
The LHV-LHS decomposition of  unsteerable white noise-BB84 box cannot be realized with hidden variable having dimension $3$ or $2$ for the whole range $ V \leq \frac{1}{\sqrt{2}}$.
\end{thm}
\begin{proof}
Note that the noisy BB84 box corresponds to the following joint probabilities:
 \begin{equation}      
 \label{matrix form BB84}
 P_{BB84}(ab|xy) =\begin{tabular}{c|cccc}
 \backslashbox{(x,y)}{(a,b)} & (0,0) & (0,1) & (1,0) & (1,1)\\\hline
(0,0) & $\frac{1+V}{4}$ & $\frac{1-V}{4}$ & $\frac{1-V}{4}$ & $\frac{1+V}{4}$ \\
(0,1) & $\frac{1}{4}$ & $\frac{1}{4}$ & $\frac{1}{4}$ & $\frac{1}{4}$ \\
(1,0) & $\frac{1}{4}$ & $\frac{1}{4}$ & $\frac{1}{4}$ & $\frac{1}{4}$ \\
(1,1) & $\frac{1-V}{4}$ & $\frac{1+V}{4}$ & $\frac{1+V}{4}$ & $\frac{1-V}{4}$ \\
\end{tabular},
\end{equation}
where each row and column corresponds to a fixed measurement setting $(xy)$ and a fixed outcome $(ab)$ respectively \cite{not}.
%
%
%
The marginal probabilities for Alice's and Bob's side are
\begin{equation}
\label{marginal}
p(a|x) =  \frac{1}{2} \hspace{0.4cm} \forall a,x
\end{equation}
and
\begin{equation}
\label{marginalB}
p(b|y) = \frac{1}{2} \hspace{0.4cm} \forall b,y
\end{equation}
respectively.\\

Now, let us try to construct a LHV-LHS decomposition of noisy BB84 box which requires a hidden variable of dimension $3$. Henceforth, we will denote a LHV-LHS decomposition of an unsteerable correlation having \textit{different deterministic probability distributions} at Alice's side and non-deterministic probability distributions (with quantum realization) at Bob's side, simply, by the term ``DLHV-LHS decomposition". Note that In $2 \times 2 \times 2$ Bell-scenario, hidden variable with dimnesion $d_{\lambda} \leq 4$ is sufficient for reproducing any local correlation \cite{sl1}. Since unsteerable correlations form a subset of the local correlations, in $2 \times 2 \times 2$ steering-scenario hidden variable with dimnesion $d_{\lambda} \leq 4$ is sufficient for reproducing any unsteerable correlation. Hence, constructing a LHV-LHS decomposition of noisy BB84 box with hidden variable dimension $3$ can be realized in the following two possible ways:\\

i) One has to construct a DLHV-LHS decomposition of noisy BB84 box (for $V \leq \frac{1}{\sqrt{2}}$) with a hidden variable of dimension $4$  as in Eq.(\ref{bblhvlhs}). Then taking equal non-deterministic distributions at Bob's side as common and making the corresponding probability distributions at Alice's side non-deterministic can reduce the dimension of the hidden variable to $3$. However, all the non-deterministic probability distributions (with quantum realization) at Bob's side $P(b|y, \rho_{\lambda})$ ($\lambda=0,1,2,3$) in the decomposition (\ref{bblhvlhs}) are unequal. In fact it can be easily checked that it is impossible to construct a DLHV-LHS decomposition of noisy BB84 box for the whole range $V \leq \frac{1}{\sqrt{2}}$ with a hidden variable of dimension $4$ with some/all non-deterministic probability distributions at Bob's side being equal to each other. Hence, the dimension of the hidden variable cannot be reduced from $4$ to $3$ in the DLHV-LHS decomposition of noisy BB84 box (for $V \leq \frac{1}{\sqrt{2}}$).\\

ii) One has to construct a DLHV-LHS decomposition of noisy BB84 box (for $V \leq \frac{1}{\sqrt{2}}$) with a hidden variable of dimension $3$. In the following we will check such possibility.\\

In this case the noisy BB84 box can be decomposed in the following way:
\begin{equation}
P_{BB84}(ab|xy) = \sum_{\lambda=0}^{2} p(\lambda) P(a|x, \lambda) P(b|y, \rho_{\lambda}).
\end{equation}
Here, $p(0)= q$, $p(1) = r$, $p(2) = s$ ($0 <q<1$, $0 <r<1$, $0 <s<1$, $q+r+s =1$). Since Alice's strategy is deterministic one, the three probability distributions $P(a|x, \lambda)$ $(\lambda = 0, 1, 2)$ must be equal to any three among $P_D^{00}$, $P_D^{01}$, $P_D^{10}$ and $P_D^{11}$. But any such combination will not satisfy the marginal probabilities for Alice given by Eq. (\ref{marginal}). So it is impossible to construct a DLHV-LHS decomposition of noisy BB84 box (for $V \leq \frac{1}{\sqrt{2}}$) with a hidden variable of dimension $3$.\\

Hence, one can conclude that it is impossible to construct a LHV-LHS decomposition of noisy BB84 box (for $V \leq \frac{1}{\sqrt{2}}$) with a hidden variable of dimension $3$ shared between Alice and Bob with deterministic/non-deterministic probability distributions at Alice's side and non-deterministic probability distributions (with quantum realization) at Bob's side.\\

Now, let us try to construct a LHV-LHS decomposition of noisy BB84 box which requires a hidden variable of dimension $2$. Since in $2 \times 2 \times 2$ steering-scenario hidden variables with dimnesions $d_{\lambda} \leq 4$ is sufficient for reproducing any unsteerable correlation, this can be realized in the following three possible ways:\\

i) One has to construct a DLHV-LHS decomposition of noisy BB84 box (for $V \leq \frac{1}{\sqrt{2}}$) with a hidden variable of dimension $4$ as in Eq.(\ref{bblhvlhs}). Then taking equal non-deterministic distributions at Bob's side as common and making the corresponding probability distributions at Alice's side non-deterministic can reduce the dimension of the hidden variable to $2$. However, as mentioned earlier, all the non-deterministic probability distributions at Bob's side $P(b|y, \rho_{\lambda})$ ($\lambda=0,1,2,3$) in the decomposition (\ref{bblhvlhs}) are unequal. In fact it can be easily checked that it is impossible to construct a DLHV-LHS decomposition of noisy BB84 box for the whole range $V \leq \frac{1}{\sqrt{2}}$ with a hidden variable of dimension $4$ with some/all non-deterministic probability distributions at Bob's side being equal to each other. Hence, the dimension of the hidden variable cannot be reduced from $4$ to $2$ in the DLHV-LHS decomposition of noisy BB84 box (for $V \leq \frac{1}{\sqrt{2}}$).\\

ii) One has to construct a DLHV-LHS decomposition of noisy BB84 box (for $V \leq \frac{1}{\sqrt{2}}$) with a hidden variable of dimension $3$. Then adapting the aforementioned procedure one can reduce the dimension of the hidden variable to $2$. However, it has already been shown that it is impossible to construct a DLHV-LHS decomposition of noisy BB84 box (for $V \leq \frac{1}{\sqrt{2}}$) with a hidden variable of dimension $3$.\\

iii) One has to construct a DLHV-LHS decomposition of noisy BB84 box (for $V \leq \frac{1}{\sqrt{2}}$) with a hidden variable of dimension $2$. In the following we will check such possibility.\\

In this case the noisy BB84 box can be decomposed in the following way:
\begin{equation}
P_{BB84}(ab|xy) = \sum_{\lambda=0}^{1} p(\lambda) P(a|x, \lambda) P(b|y, \rho_{\lambda}).
\end{equation}
Here, $p(0)=q$, $p(1)=r$ ($0 <q<1$, $0 <r<1$, $q+r =1$). Since Alice's strategy is a deterministic one, the two probability distributions $P(a|x, \lambda)$ $(\lambda = 0, 1)$ must be equal to any two among $P_D^{00}$, $P_D^{01}$, $P_D^{10}$ and $P_D^{11}$. In order to satisfy the marginal probabilities for Alice given by Eq. (\ref{marginal}), the only two possible choices of $P(a|x, 0)$ and $P(a|x, 1)$ are:\\
1) $P_D^{00}$ and $P_D^{01}$ with $q=r=\frac{1}{2}$\\
2) $P_D^{10}$ and $P_D^{11}$ with $q=r=\frac{1}{2}$.\\

\underline{1st Choice}\\
Now consider the first choice, i.e., $P(a|x, 0)$ =  $P_D^{00}$ and $P(a|x, 1)$ = $P_D^{01}$ (with $q=r=\frac{1}{2}$). Now in order to satisfy the 1st and 4th row given in Eq.(\ref{matrix form BB84}), the only possible choice for $P(b|y, \rho_{0})$ and $P(b|y, \rho_{1})$ are:
\begin{eqnarray} 
P(b|y,\rho_0) = \begin{tabular}{c|cc}
\backslashbox{(y)}{(b)} & (0) & (1) \\\hline
(0) & $\frac{1+V}{2}$ & $\frac{1-V}{2}$  \\
(1) & $\frac{1-V}{2}$ & $\frac{1+V}{2}$ \\
\end{tabular}  =  \langle \psi _0 | \{\Pi_{b|y}\}_{b,y} | \psi_0 \rangle
\end{eqnarray}
and
\begin{eqnarray}
P(b|y,\rho_1) = \begin{tabular}{c|cc}
\backslashbox{(y)}{(b)} & (0) & (1) \\\hline
(0) & $\frac{1-V}{2}$ & $\frac{1+V}{2}$  \\
(1) & $\frac{1+V}{2}$ & $\frac{1-V}{2}$ \\
\end{tabular}  =  \langle \psi _1 | \{\Pi_{b|y}\}_{b,y} | \psi_1 \rangle.
\end{eqnarray}
In this case, the marginal probabilities for Bob given by Eq.(\ref{marginalB}) are satisfied. But the joint probabilities given  in the 2nd and 3rd row of  Eq.(\ref{matrix form BB84}) are not satisfied.
In a similar way, it can be shown that, in case of the first choice, if one wants to satisfy  the 2nd and 3rd row in Eq.(\ref{matrix form BB84}), then the marginal probabilities for Bob given by Eq.(\ref{marginalB}) will be satisfied, but the 1st and 4th row in Eq.(\ref{matrix form BB84}) will not be satisfied.
In this way it can be shown that with the choice $P(a|x, 0)$ =  $P_D^{00}$ and $P(a|x, 1)$ = $P_D^{01}$, all joint probabilities cannot be satisfied simultaneously.

\underline{2nd Choice}\\
Now consider the second choice, i.e., $P(a|x, 0)$ =  $P_D^{10}$ and $P(a|x, 1)$ = $P_D^{11}$ (with $q=r=\frac{1}{2}$). In order to satisfy  the 1st and 4th row given in Eq.(\ref{matrix form BB84}), the only possible choice for $P(b|y, \rho_{0})$ and $P(b|y, \rho_{1})$ are:
\begin{eqnarray} 
P(b|y,\rho_0) = \begin{tabular}{c|cc}
\backslashbox{(y)}{(b)} & (0) & (1) \\\hline
(0) & $\frac{1+V}{2}$ & $\frac{1-V}{2}$  \\
(1) & $\frac{1+V}{2}$ & $\frac{1-V}{2}$ \\
\end{tabular}   =  \langle \psi _2 | \{\Pi_{b|y}\}_{b,y} | \psi_2 \rangle 
\end{eqnarray}
and
\begin{eqnarray}
P(b|y,\rho_1) = \begin{tabular}{c|cc}
\backslashbox{(y)}{(b)} & (0) & (1) \\\hline
(0) & $\frac{1-V}{2}$ & $\frac{1+V}{2}$  \\
(1) & $\frac{1-V}{2}$ & $\frac{1+V}{2}$ \\
\end{tabular}  =  \langle \psi _3 | \{\Pi_{b|y}\}_{b,y} | \psi_3 \rangle.
\end{eqnarray}
In this case, the marginal probabilities for Bob given by Eq.(\ref{marginalB}) are satisfied. But the joint probabilities given in the 2nd and 3rd row of Eq.(\ref{matrix form BB84}) are not satisfied.
In a similar way, it can be shown that, in case of the second choice, if one wants to satisfy the 2nd and 3rd row of Eq.(\ref{matrix form BB84}), then the marginal probabilities for Bob given by Eq.(\ref{marginalB}) will be satisfied. But the 1st and 4th row of Eq.(\ref{matrix form BB84}) will not be satisfied.
In this way it can be shown that with the choice $P(a|x, 0)$ =  $P_D^{10}$ and $P(a|x, 1)$ = $P_D^{11}$, all joint probabilities cannot be satisfied simultaneously.  It is, therefore, impossible to construct a DLHV-LHS decomposition of noisy BB84 box (for $V \leq \frac{1}{\sqrt{2}}$) with a hidden variable of dimension $2$.\\

Hence, we can conclude that it is impossible to have a LHV-LHS decomposition of noisy BB84 box (for $V \leq \frac{1}{\sqrt{2}}$) with a hidden variable of dimension $2$ shared between Alice and Bob having deterministic/non-deterministic probability distributions at Alice's side and non-deterministic probability distributions (with quantum realization) at Bob's side.\\
\end{proof}
The above theorem implies the following. 
\begin{cor}
 The unsteerable white-noise BB84 family demonstrates super-unsteerablity.
\end{cor}
\begin{proof}
We have shown that the unsteerable white noise-BB84 family can have LHV-LHS model with the minimum dimension of the hidden variable being $4$ for the whole range $ V \leq \frac{1}{\sqrt{2}}$. On the other hand, we have seen that this white noise-BB84 family can be simulated by using $2 \otimes 2$ quantum system (\ref{w}).  This is an instance of super-unsteerability since the minimum dimension of shared randomness needed for simulating the LHV-LHS model of unsteerable white-noise BB84 family is greater than the local Hilbert space dimension of the shared quantum system (reproducing unsteerable white-noise BB84 family) at the untrusted party's side (who steers the other party, in the present case Bob).
\end{proof}

As discussed before, in $2 \times 2 \times 2$ steering-scenario hidden variables with dimnesions $d_{\lambda} \leq 4$ is sufficient for reproducing any unsteerable correlation. Previously we have shown an example of super-unsteerability in the $2 \times 2 \times 2$ experimental scenario where the classical simulation protocol (with LHV-LHS model) requires hidden variables having minimum dimension $4$. Hence, there may be another form super-unsteerability where the classical simulation protocol (with LHV-LHS model) requires minimum hidden variable dimension $3$. In the following subsection we are going to present an example of it.\\

\subsection{Super-unsteerability: Example 2}

Consider that the two spatially separated parties (say, Alice and Bob) share the following separable two-qubit state,
\begin{equation}
\label{state2n}
\rho = \frac{1}{2} \Big( |00\rangle \langle 00| + |++ \rangle \langle ++| \Big) ,
\end{equation}
where, $|0\rangle$ and $|+\rangle$ are the eigenstates of the operators $\sigma_z$ and $\sigma_x$, respectively, correspondng to the eigenvalue $+1$. The above state has nonzero  quantum discord  from both
Alice to Bob and Bob to Alice since it  is neither a classical-quantum state nor a quantum-classical state \cite{disc, disc2, disc3, disc4}. If Alice performs the projective measurements of observables corresponding to the operators $A_0 = \sigma_x$ and $A_1 = \sigma_z$, and Bob performs projective measurements of observables corresponding to the operators $B_0 =  \sigma_x$ and $B_1 = \sigma_z$, then the following correlation is produced from the above quantum-quantum state,

 \begin{equation}      
 \label{corr2}
 P(ab|xy) =\begin{tabular}{c|cccc}
 \backslashbox{(x,y)}{(a,b)} & (0,0) & (0,1) & (1,0) & (1,1)\\\hline\\[0.05cm]
(0,0) & $\dfrac{5}{8}$ & $\dfrac{1}{8}$ & $\dfrac{1}{8}$ & $\dfrac{1}{8}$ \\[0.5cm]
(0,1) & $\dfrac{1}{2}$ & $\dfrac{1}{4}$ & $\dfrac{1}{4}$ & $0$ \\[0.5cm]
(1,0) & $\dfrac{1}{2}$ & $\dfrac{1}{4}$ & $\dfrac{1}{4}$ & $0$ \\[0.5cm]
(1,1) & $\dfrac{5}{8}$ & $\dfrac{1}{8}$ & $\dfrac{1}{8}$ & $\dfrac{1}{8}$ \\
\end{tabular}
\end{equation}

Here $x$, $y$ denote the input variables on Alice's and Bob's sides respectively; and $a$, $b$ denote the outputs on Alice's and Bob's sides respectively. The above box  does not violate the analogous CHSH inequality for steering (\ref{chshst}). Hence, the box (\ref{corr2}) is unsteerable in the scenario where Alice performs black-box (uncharacterized) measurements and Bob performs two mutually unbiased qubit measurements. In the following, we demonstrate that the box (\ref{corr2}) detects super-unsteerability 
of the quantum-quantum state (\ref{state2n}).

\subsubsection*{Simulating the correlation given by Eq.(\ref{corr2}) with LHV at one side and LHS at another side}

The  correlation given by Eq.(\ref{corr2}) can be written as
\begin{equation}
P(ab|xy) = \frac{1}{8}\Big( 2 P_D^{0000} + P_D^{0010} + P_D^{0011} + P_D^{1000} + P_D^{1100} + P_D^{1010} + P_D^{1111} \Big) \nonumber
\end{equation}
\begin{equation}
= \frac{1}{2} P_D^{00} \Big( \frac{2 P_D^{00} + P_D^{10} + P_D^{11}}{4} \Big) \nonumber
\end{equation}
\begin{equation}
+ \frac{1}{4} P_D^{10} \Big( \frac{P_D^{00} + P_D^{10}}{2} \Big) + \frac{1}{4} P_D^{11} \Big( \frac{P_D^{00} + P_D^{11}}{2} \Big) \nonumber
\end{equation}
\begin{equation}
= \sum_{\lambda=0}^{2} p(\lambda) P(a|x, \lambda) P(b|y, \rho_{\lambda}),
\label{eee}
\end{equation}
where $p(0)$ = $\frac{1}{2}$, $p(1)$ = $p(2)$ = $\frac{1}{4}$;\\
$P(a|x,0)$ = $P_D^{00}$, $P(a|x,1)$ = $P_D^{10}$, $P(a|x,2)$ = $P_D^{11}$;\\
and
\begin{align}      
P(b|y,\rho_0) \!=\!  \frac{2 P_D^{00} + P_D^{10} + P_D^{11}}{4} \! &=\! \begin{tabular}{c|cc}
 \backslashbox{(y)}{(b)} & (0) & (1) \\\hline
(0) & $\frac{3}{4}$ & $\frac{1}{4}$  \\
(1) & $\frac{3}{4}$ & $\frac{1}{4}$  \\
\end{tabular} \nonumber\\
&=\!  \langle \psi^{'}_0 | \{\Pi_{b|y}\}_{b,y} | \psi^{'}_0 \rangle,
\end{align}
\begin{align}      
P(b|y,\rho_1) \!=\!  \frac{P_D^{00} + P_D^{10}}{2} \! &=\! \begin{tabular}{c|cc}
 \backslashbox{(y)}{(b)} & (0) & (1) \\\hline
(0) & $1$ & $0$  \\
(1) & $\frac{1}{2}$ & $\frac{1}{2}$  \\
\end{tabular} \nonumber\\
&=\!  \langle \psi^{'}_1 | \{\Pi_{b|y}\}_{b,y} | \psi^{'}_1 \rangle,
\end{align}
\begin{align}      
P(b|y,\rho_2) \!=\!  \frac{P_D^{00} + P_D^{11}}{2} \! &=\! \begin{tabular}{c|cc}
 \backslashbox{(y)}{(b)} & (0) & (1) \\\hline
(0) & $\frac{1}{2}$ & $\frac{1}{2}$  \\
(1) & $1$ & $0$  \\
\end{tabular} \nonumber\\
&=\!  \langle \psi^{'}_2 | \{\Pi_{b|y}\}_{b,y} | \psi^{'}_2 \rangle,
\end{align}
 where $\{\Pi_{b|y}\}_{b,y}$ corresponds to two arbitrary projective mutually unbiased measurements in the Hilbert space $\mathcal{C}^2$ corresponding to the operators $B_0 = |\uparrow_0 \rangle \langle \uparrow_0|$ $-$ $|\downarrow_0 \rangle \langle \downarrow_0|$ and $B_1 =  |\uparrow_1 \rangle \langle \uparrow_1|$ $-$ $|\downarrow_1 \rangle \langle \downarrow_1|$ as described earlier. The $|\psi^{'}_{\lambda}\rangle$s that produce
$p(b|y,\rho_{\lambda})$s given above are given by  
\begin{equation}
|\psi^{'}_0 \rangle = \frac{\sqrt{3}}{2} |\uparrow_0 \rangle + e^{i \phi^{'}_0} \frac{1}{2} |\downarrow_0 \rangle,
\end{equation}
where $cos \phi^{'}_0 = \frac{1}{\sqrt{3}}$,
\begin{equation}
|\psi^{'}_1 \rangle =  |\uparrow_0 \rangle
\end{equation}
and
\begin{equation}
|\psi^{'}_2 \rangle = \frac{1}{\sqrt{2}} |\uparrow_0 \rangle +  \frac{1}{\sqrt{2}} |\downarrow_0 \rangle,
\end{equation}
which are all valid states in the Hilbert space $\mathcal{C}^2$.\\

Hence, the LHV-LHS decomposition of  the correlation given by Eq.(\ref{corr2}) can be realized with hidden variable having dimension $3$ (with two arbitrary projective mutually unbiased measurements at trusted party).

\begin{thm}
The LHV-LHS decomposition of  the correlation given by Eq.(\ref{corr2}) cannot be realized with hidden variable having dimension $2$.
\end{thm}
\begin{proof}
Let us try to construct a LHV-LHS decomposition of the correlation given by Eq.(\ref{corr2}) which requires a hidden variable of dimension $2$.  Since in $2 \times 2 \times 2$ steering-scenario hidden variables with dimnesions $d_{\lambda} \leq 4$ is sufficient for reproducing any unsteerable correlation, this can be realized in the following three possible ways:\\

i) One has to construct a DLHV-LHS decomposition of the correlation given by Eq.(\ref{corr2}) with a hidden variable of dimension $4$. Then taking equal non-deterministic distributions at Bob's side as common and making the corresponding probability distributions at Alice's side non-deterministic can reduce the dimension of the hidden variable to $2$. However, it can be shown that it is impossible to to construct a DLHV-LHS decomposition of the correlation given by Eq.(\ref{corr2}) with a hidden variable of dimension $4$ (for detailed calculations, see the Appendix).\\

ii) One has to construct a DLHV-LHS decomposition of the correlation given by Eq.(\ref{corr2}) with a hidden variable of dimension $3$ as in Eq.(\ref{eee}). Then by taking equal non-deterministic distributions at Bob's side as common and making the corresponding probability distributions at Alice's side non-deterministic one can reduce the dimension of the hidden variable to $2$. However, all the non-deterministic probability distributions at Bob's side $P(b|y, \rho_{\lambda})$ ($\lambda=0,1,2,3$) in the decomposition (\ref{eee}) are unequal. In fact it can be easily checked that it is impossible to construct a DLHV-LHS decomposition of the correlation (\ref{corr2}) with a hidden variable of dimension $3$ with some/all non-deterministic probability distributions at Bob's side being equal to each other. Hence, the dimension of the hidden variable cannot be reduced from $3$ to $2$ in the DLHV-LHS decomposition of the correlation (\ref{corr2}).\\

iii) One has to construct a DLHV-LHS decomposition of the correlation (\ref{corr2}) with a hidden variable of dimension $2$. In the following we will check such possibility.\\

 In this case the correlation given by Eq.(\ref{corr2}) can be decomposed in the following way:
\begin{equation}
P(ab|xy) = \sum_{\lambda=0}^{1} p(\lambda) P(a|x, \lambda) P(b|y, \rho_{\lambda}).
\end{equation}
Here, $p(0)=q$, $p(1)=r$ ($0 <q<1$, $0 <r<1$, $q+r =1$). Since Alice's strategy is a deterministic one, the two probability distributions $P(a|x, \lambda)$ $(\lambda = 0, 1)$ must be equal to any two among $P_D^{00}$, $P_D^{01}$, $P_D^{10}$ and $P_D^{11}$. But it can be easily checked that none of these choices will satisfy all the joint probability distributions mentioned in Eq.(\ref{corr2}) simultaneously.
It is, therefore, impossible to to construct a DLHV-LHS decomposition of the correlation (\ref{corr2}) with a hidden variable of dimension $2$.\\

Hence, we can conclude that it is impossible to have a LHV-LHS decomposition of the correlation (\ref{corr2}) with a hidden variable of dimension $2$ shared between Alice and Bob having deterministic/non-deterministic probability distributions at Alice's side and non-deterministic probability distributions (with quantum realization) at Bob's side.\\
\end{proof}
The above theorem implies the following. 
\begin{cor}
The correlation given by Eq.(\ref{corr2}) demonstrates super-unsteerablity. 
\end{cor}
\begin{proof}
We have shown that the unsteerable correlation given by Eq.(\ref{corr2}) can have a LHV-LHS model
with the minimum dimension of the hidden variable being $3$. On the other hand, we have seen that the unsteerable correlation given by Eq.(\ref{corr2}) can be simulated by using $2 \otimes 2$ quantum system (\ref{state2n}). This is an instance of super-unsteerability since the minimum dimension of shared randomness needed to simulate the LHV-LHS model of the correlation (\ref{corr2}) is greater than the local Hilbert space dimension of the shared quantum system (reproducing the given unsteerable correlation) at the untrusted party's side (who steers the other party, in the present case Bob).
\end{proof}


\section{Quantumness as captured by super-unsteerability}
Here we argue that the unsteerable boxes (\ref{bb84}) and (\ref{corr2}) in 
the given steering scenario, where the dimension of the steering party is bounded, have nonclassicality beyond steering which 
can be operationally identified with super-unsteerability.\\

 In the Example $1$ of super-unsteerability, we have shown that the unsteerable BB84 box (for $ V \leq \frac{1}{\sqrt{2}}$) given by Eq.(\ref{bb84}) can be simulated by LHV-LHS model with random variable having minimum dimension $4$, where each LHS is a $2$ dimensional quantum system. In another words,  the unsteerable BB84 box (for $ V \leq \frac{1}{\sqrt{2}}$) given by Eq.(\ref{bb84}) cannot be reproduced by classical-quantum state with dimension $d \otimes 2$ where $d < 4$. Hence,  the super-unsteerable BB84 box (for $ V \leq \frac{1}{\sqrt{2}}$) cerifies quantumness of $d \otimes 2$  dimensional resources producing it, where $d <4$. For example, the super-unsteerable BB84 box certifies the quantumness of the $2 \otimes 2$ states given by Eq.(\ref{w}) for $ V \leq \frac{1}{\sqrt{2}}$.\\
 
Consider that Alice and Bob share the following qutrit-qubit state:
\begin{equation}
\label{erasure}
\rho^{'}_V = V |\psi^- \rangle \langle \psi^- | + \frac{1- V}{2} |2 \rangle \langle 2 | \otimes \mathbb{I}_2,
\end{equation}
where $|\psi^- \rangle$ = $\frac{1}{\sqrt{2}} (|01 \rangle - |10 \rangle)$ is the singlet state, $\mathbb{I}_2$ is the identity in the $|0\rangle$, $|1\rangle$ qubit subspace and $0 < V \leq 1$. This state is known as the ``Erasure state'', as it can be obtained by sending Alice's qubit of a bipartite singlet state through an erasure channel; with probability $V$ the singlet state remains the same, and with $(1-V)$ probability Alice's qubit is replaced by the state $|2 \rangle \langle 2 |$ (orthogonal to the qubit subspace). This state is entangled for any non-zero $V$ which can be checked
through the positive-partial-transpose criterion \cite{ppt}. Therefore, the Erasure state $\rho^{'}_V$ has nonzero quantumness (as quantified by quantum discord \cite{disc, disc2, disc3, disc4}) for any $V>0$.
If Alice and Bob perform appropriate measurement on the Erasure state, the white noise-BB84 family can also  be reproduced (for detailed calculations, see the Appendix). Hence, the super-unsteerable BB84 box certifies the quantumness of the $3 \otimes 2$ states given by Eq.(\ref{erasure}) for $ V \leq \frac{1}{\sqrt{2}}$.\\

Thus, super-unsteerability provides operational characterization of quantumness of the unsteerable (in the given steering scenario) $2 \otimes 2$ state (\ref{w}) and $3 \otimes 2$  state (\ref{erasure}) for $V \leq 1/\sqrt{2}$ if the local Hilbert-space dimension of the steering party is bounded.\\

In Example 2 we have considered a  $2 \otimes 2$ separable mixed state (\ref{state2n}) having
both nonzero Alice to Bob and non-zero Bob to Alice discord.
We have shown that the correlation (\ref{corr2}) produced by performing some particular local noncommuting measurements on this state can be simulated with LHV-LHS model with random variable having minimum dimension $3$. This implies that the correlation (\ref{corr2}) cannot be simulated by classical-quantum state with dimension $2 \otimes 2$. Hence, the super-unsteerable correlation (\ref{corr2}) certifies quantumness of certain $2 \otimes 2$ dimensional resources producing it and provides operational characterization of quantumness of the state (\ref{state2n}).\\

Consider  the classical-quantum or  quantum-classical  states \cite{bc} given by,
\begin{equation}
\rho_{CQ} = \sum_{i=0}^{1} p_i |i\rangle\langle i| \otimes \chi_i
\label{eq:cq}
\end{equation}
and
\begin{equation}
\label{cq:eq}
\rho_{QC} = \sum_{j=0}^{1} p_{j}  \phi_j  \otimes | j \rangle \langle j|,
\end{equation}
where  $ \{  |  i  \rangle \}$  and  $  \{ |  j  \rangle  \}$ are  the
orthonormal sets, and, $\chi_i$ and $\phi_j$ are the arbitrary quantum
states.    These correlations  can  manifestly  be simulated  by
  presharing randomness of dimension  2.  Since Eqs. (\ref{eq:cq}) and
  (\ref{cq:eq}) represent a  family of states that  are not super-correlated,
  they  cannot be used to demonstrate super-unsteerability.  The zero-discord
states  having  classical-classical   correlations  (corresponding  to
orthogonal $\chi_i$ and $\phi_j$ in Eqs.(\ref{eq:cq}, \ref{cq:eq}) respectively) are also not
super-unsteerable.     Therefore   one    can   conclude    that   the
super-unsteerable   states   form   a    subset   of   states   having
\textit{quantum-quantum} correlations,  and thus form a  strict subset
of the discordant states. Further, in the $2 \times 2 \times 2$ experimental scenario super-unsteerability classifies any bipartite states having quantum-quantum correlations into three types:\\

(i) quantum-quantum states which demonstrate super-unsteerability with unsteerable boxes having 
shared randomness of minimum dimension $4$.\\

(ii) quantum-quantum states which demonstrate super-unsteerability with unsteerable boxes having 
shared randomness of minimum dimension $3$ and \\

(iii) quantum-quantum states which do not demonstrate super-unsteerability.

\section{Inequivalence between superlocality and super-unsteerability}
 
 Superlocality  \cite{sl1, sl2, sl3, sl4, sl5, sl6} of bipartite quantum correlations, which do not violate a Bell inequality, refers to the higher dimensionality of shared randomness needed to reproduce them in the classical simulation scenarios compared to that of the quantum systems producing the correlations. In this classical simulation scenario the local hidden variables are used by both the parties to generate the shared randomness.  On the other hand, super-unsteerability of bipartite quantum correlations, which do not violate a steering inequality, refers to the higher dimensionality of shared randomness needed to reproduce them compared to the local Hilbert space dimension of the quantum system (reproducing the correlation) at the untrusted party's side. This notion is defined in the classical simulation scenarios where local hidden variables are used by one of the parties and the other party uses local hidden states. Thus, the classical simulation scenarios in which superlocality and super-unsteerability  defined are completely inequivalent; the former corresponds to the 
black-box models on both the sides as it corresponds to the Bell scenario, whereas the latter corresponds to the black-box model on one side and the quantum model on the other side as it corresponds to the steering scenario. \\

For instance, the BB84 family is superlocal  for $0 < V \leq 1$ (see Proposition $3$ in Ref. \cite{sl1} where superlocality of the BB84 family with $V=1$ was shown) and  detects   steerability for $V>1/\sqrt{2}$.  On the other hand, it detects super-unsteerability for $0 < V \leq \frac{1}{\sqrt{2}}$. Hence, these two different regions where the BB84 family is superlocal and super-unsteerable demonstrate the inequivalence between superlocality and super-unsteerability.

 \section{Discussion and Conclusion}

In this work we have introduced the notion of super-unsteerability by showing that there are certain unsteerable  correlations whose simulation with LHV-LHS model requires preshared randomness with dimension higher than the local Hilbert space dimension of the quantum system (reproducing the given unsteerable correlation) at the untrusted party's side. Two examples of super-unsteerability has been presented. These two examples are inequivalent from each other with respect to minimum dimension of the shared randomness required for simulating the correlations with LHV-LHS models.  Note that in the present study we have restricted ourselves to the $2 \times 2 \times 2$ experimental scenario ($2$ parties, $2$ measurement settings per party, $2$ outcomes per measurement setting), and in this scenario shared randomness with dimension $d_{\lambda} \leq 4$ is sufficient to simulate any local as well as unsteerable correlation \cite{sl1}. Hence, there are two possible classes of super-unsteerability in the $2 \times 2 \times 2$ experimental scenario: one with shared randomness with minimum dimension of $4$ and the other with shared randomness with minimum dimension of $3$. We have presented examples of both this possible two classes of super-unsteerability in the $2 \times 2 \times 2$ experimental scenario.
Further, our study of simulating unsteerable boxes with minimum dimension of the shared randomness in the $2 \times 2 \times 2$ experimental scenario classifies any bipartite states into 
three types: (i) States which do not demonstrate super-unsteerability. The classical-quantum and quantum-classical states belong to this class.
(ii) quantum-quantum states which demonstrate super-unsteerability with unsteerable boxes having minimum hidden variable dimension $3$,
and (iii) quantum-quantum states which demonstrate super-unsteerability with unsteerable boxes having minimum hidden variable dimension $4$. The present study also provides a efficient procedure to minimize the dimension of the shared randomness needed to construct the LHV-LHS model of an unsteerable correlation.\\

 In   Ref.   \cite{sl5},  the   authors  have  shown   that  the
  nonclassicality of a  family of local correlations  in the Bell-CHSH
  scenario  can be characterized by  super-correlation, in this
  case, superlocality.   Extending this  approach, we have shown here
  that the nonclassicality of certain unsteerable states in the related steering  scenario can also
  be pointed out by  super-correlations,
  i.e., the phenomenon  of super-unsteerability.  \\

Before concluding, we note that nonlocality or  steerability of any  correlation in QM or in any convex operational theory can be characterized by the non-zero communication cost that must be supplemented with pre-shared randomness in order to simulate the correlations. The question of an analogous operational characterization of quantumness of unsteerable correlations has been addressed here, and associated with super-unsteerability. The idea of super-unsteerability in the context of multipartite unsteerable boxes \cite{munbox, munbox2} would be worth probing for future studies. It would be interesting to study how to quantify super-unsteerability and whether there exists any quantum informational application of the quantumness of unsteerable correlations as witnessed by super-unsteerability.

\section{ACKNOWLEDGEMENTS}
DD acknowledges the financial support from University Grants Commission (UGC), Government of India.
BB, CJ and ASM acknowledge the financial support from  project SR/S2/LOP-08/2013 of the Department of Science and Technology (DST), government
of India.

\appendix
\section{Demontrating that the correlation given by Eq.(\ref{corr2}) cannot have a DLHV-LHS decomposition with a hidden variable of dimension $4$}	

If the correlation given by Eq.(\ref{corr2}) has a DLHV-LHS decomposition with a hidden variable of dimension $4$, i. e., a LHV-LHS decomposition with a hidden variable of dimension $4$ having different deterministic probability distributions at Alice's side and non-deterministic probability distributions (with quantum realization) at Bob's side, then the correlation (\ref{corr2}) can be written as follows,
\begin{equation}
\label{a1}
P_{BB84}(ab|xy) = \sum_{\lambda=0}^{3} p(\lambda) P(a|x, \lambda) P(b|y, \rho_{\lambda}),
\end{equation}
where, $0<p(\lambda) <1$ ($\lambda = 0, 1, 2, 3$), $P(a|x,0) = P_D^{00}$, $P(a|x,1) = P_D^{01}$, $P(a|x,2) = P_D^{10}$, $P(a|x,3) = P_D^{11}$. Let us assume that the non-deterministic probability distributions (with quantum realization) at Bob's side are given by,
\begin{equation}      
P(b|y,\rho_0) =\begin{tabular}{c|cc}
 \backslashbox{(y)}{(b)} & (0) & (1) \\\hline
(0) & $t_1$ & $1-t_1$  \\
(1) & $t_2$ & $1-t_2$  \\
\end{tabular}
\end{equation}
\begin{equation}      
P(b|y,\rho_1) =\begin{tabular}{c|cc}
 \backslashbox{(y)}{(b)} & (0) & (1) \\\hline
(0) & $t_3$ & $1-t_3$  \\
(1) & $t_4$ & $1-t_4$  \\
\end{tabular}
\end{equation}
\begin{equation}      
P(b|y,\rho_2) =\begin{tabular}{c|cc}
 \backslashbox{(y)}{(b)} & (0) & (1) \\\hline
(0) & $t_5$ & $1-t_5$  \\
(1) & $t_6$ & $1-t_6$  \\
\end{tabular}
\end{equation}
\begin{equation}      
P(b|y,\rho_3) =\begin{tabular}{c|cc}
 \backslashbox{(y)}{(b)} & (0) & (1) \\\hline
(0) & $t_7$ & $1-t_7$  \\
(1) & $t_8$ & $1-t_8$  \\
\end{tabular}.
\end{equation}
with $0 < t_i < 1$ ($i=1,2,3,4,5,6,7,8$). Comparing Eqs.(\ref{a1}) and (\ref{corr2}), we get,
\begin{equation}
\label{a2}
p(11|01) = p(1)(1-t_4) + p(3)(1-t_8) =0
\end{equation}
and
\begin{equation}
\label{a3}
p(11|10) = p(1)(1-t_3) + p(2)(1-t_5) =0.
\end{equation}
Since $0<p(\lambda) <1$ ($\lambda = 0, 1, 2, 3$) and $0 < t_i < 1$ ($i=1,2,3,4,5,6,7,8$), from Eqs.(\ref{a2}) and (\ref{a3}), we get,
\begin{equation}
\label{a4}
p(1)(1-t_4) = p(1)(1-t_3) = 0.
\end{equation}
Hence, we get either,
\begin{equation}
\label{a5}
p(1)=0
\end{equation}
or,
\begin{equation}
t_4 =1 \hspace{0.4cm} \text{and} \hspace{0.4cm} t_3 =1.
\end{equation}
Now, if $p(1) =0$, then the decomposition (\ref{a1}) becomes DLHV-LHS decomposition of the correlation (\ref{corr2}) with a hidden variable of dimension $3$. On the other hand, if $t_4 =1$ and $t_3 =1$, then $P(b|y, \rho_1)$ becomes
\begin{equation}      
P(b|y,\rho_1) =\begin{tabular}{c|cc}
 \backslashbox{(y)}{(b)} & (0) & (1) \\\hline
(0) & $1$ & $0$  \\
(1) & $1$ & $0$  \\
\end{tabular},
\end{equation}
which has no quantum realisation, i. e., $P(b|y, \rho_1) \neq \langle \phi | \{\Pi_{b|y}\}_{b,y}|\phi\rangle$ for any quantum state $|\phi\rangle$.\\

Hence, one can conclude that the correlation given by Eq.(\ref{corr2}) cannot have a DLHV-LHS decomposition with a hidden variable of dimension $4$.

\section{Reproducing white noise-BB84 box using qutrit-qubit system}	
Consider Alice and Bob share the following qutrit-qubit state:
\begin{equation}
\rho_E = V |\psi^- \rangle \langle \psi^- | + \frac{1- V}{2} |2 \rangle \langle 2 | \otimes \mathbb{I}_2,
\end{equation}
where $|\psi^- \rangle$ = $\frac{1}{\sqrt{2}} (|01 \rangle - |10 \rangle)$ is the singlet state; $0 <V \leq 1$; $|0\rangle$, $|1\rangle$ and $|2\rangle$ form an orthonormal basis in the Hilbert space in $\mathcal{C}^3$;  $|0\rangle$ and $|1\rangle$ form an orthonormal basis in the Hilbert space in $\mathcal{C}^2$ (they are eigenvectors of the operator $\sigma_z$); $\mathbb{I}_2 = |0\rangle \langle 0| + |1\rangle \langle 1|$.\\

Now consider the following two dichotomic POVM $E^1 \equiv \{ E_i^1 (i=0,1) | \sum_i E_i^1 = \mathbb{I}, 0 < E_i^1 \leq \mathbb{I} \}$ and $E^2 \equiv \{ E_j^2 (j=0,1) | \sum_j E_j^2 = \mathbb{I}, 0 < E_j^2 \leq \mathbb{I} \}$, where\\
$E_0^1 = \begin{pmatrix}
0 && 0 && 0 \\
0 && 1 && 0 \\
0 && 0 && \frac{1}{2} \\
\end{pmatrix}$ and let us assume that the corresponding outcome is $0$,\\

$E_1^1 = \begin{pmatrix}
1 && 0 && 0 \\
0 && 0 && 0 \\
0 && 0 && \frac{1}{2} \\
\end{pmatrix}$ and let us assume that the corresponding outcome is $1$. \\

On the other hand,\\

$E_0^2 = \begin{pmatrix}
\frac{1}{2} && \frac{1}{2} && 0 \\
\frac{1}{2} && \frac{1}{2} && 0 \\
0 && 0 && \frac{1}{2} \\
\end{pmatrix}$ and let us assume that the corresponding outcome is $0$, \\

$E_1^2 = \begin{pmatrix}
\frac{1}{2} && -\frac{1}{2} && 0 \\
-\frac{1}{2} && \frac{1}{2} && 0 \\
0 && 0 && \frac{1}{2} \\
\end{pmatrix}$ and let us assume that the corresponding outcome is $1$, \\

Here, the matrix $E_0^1$, $E_1^1$, $E_0^2$ and $E_1^2$ are written in the basis \{$|0\rangle$, $|1\rangle$, $|2\rangle$\}. Now if  Alice performs the POVMs corresponding to $A_0 = E^1$ and $A_1 = E^2$; and Bob performs the projective measurements corresponding to $B_0 = \sigma_z$ and $B_1 = \sigma_x$, then the white noise-BB84 family can be reproduced.

\end{document}